\documentclass[%
 aip,
 jmp,%
 amsmath,amssymb,
preprint,%
]{revtex4-1}

\usepackage{graphicx}
\usepackage{bm}

\usepackage{amsthm}

\font\bb=msbm10 scaled 1200

\font\tenscr=rsfs10 scaled 1200

\newcommand{\de}{\partial}

\newcommand{\frin}[3]{{\left( \frac{#1}{#2} \right)}^{#3}}
\newcommand{\rf}[1]{(\ref{#1})} 
\newcommand{\nonu}{\nonumber}

\newcommand{\R}{\mbox {\bb R}}
\newcommand{\scal}[3]{{#1}_{#2}^{(#3)}}

\newcommand{\dst}{\displaystyle}

\newcommand{\Bs}{\bm \sigma}
\newcommand{\BE}{\bm E}
\newcommand{\Bj}{\bm j}
\newcommand{\I}{\bm I}

\newcommand{\G}{\mbox {\tenscr G}}

\newcommand{\M}{{\bm M}_{h}}
\newcommand{\se}{\sigma_{ex}}
\newcommand{\re}{\mbox{Re\,}}
\newcommand{\im}{\mbox{Im\,}}
\newcommand{\omo}{\omega_1}
\newcommand{\omt}{\omega_2}
\newcommand{\te}{\vartheta}
\newcommand{\tr}{{\sf T}}
\newcommand{\ui}{u_{in}}
\newcommand{\ue}{u_{ex}}

\renewcommand{\se}{\sigma_{ex}}
\newcommand{\si}{\sigma_{in}}
\newcommand{\sxe}{\sigma_x^\ast}
\newcommand{\sxye}{\sigma_{xy}^\ast}
\newcommand{\sye}{\sigma_y^\ast}
\newcommand{\x}{{\bm x}}
\renewcommand{\Pi}{\Phi_{in}}
\newcommand{\Pe}{\Phi_{ex}}
\renewcommand{\P}{P_{m,n}}
\renewcommand{\gamma}{\alpha}
\renewcommand{\leq}{\leqslant}
\renewcommand{\geq}{\geqslant}
\renewcommand{\iint}{\int}

\renewcommand{\phi}{\varphi}
\renewcommand{\theequation}{\thesection.\arabic{equation}}

\newtheorem{theorem}{Theorem}

\begin{document}

\title{The effective conductivity of a periodic lattice of circular inclusions}
\author{Yuri A. Godin}
\affiliation{
$^1$Department of Mathematics and Statistics,
University of North Carolina at Charlotte,
Charlotte, NC 28223, U.S.A.}
\email{ygodin@uncc.edu}

\date{\today}

\begin{abstract}
We determine the effective conductivity of a two-dimensional composite consisting
of a doubly periodic array of identical circular cylinders within a homogeneous matrix. 
The problem is reduced to the solution of an infinite system of linear equations.
The effective conductivity tensor is obtained in the form of the series expansion in terms of
the volume fraction of the cylinders whose coefficients are determined exactly.
Results are illustrated by examples.

\end{abstract}

\pacs{05.60.Cd, 41.20.Cv, 72.80.Tm, 78.20.Bh, 72.10.-d,  72.15.-v, 77.22.Ch}

\keywords{effective conductivity, overall properties, composite material, periodic media,
homogenization}
\maketitle

\section{Introduction}
\renewcommand{\theequation}{\arabic{equation}}

We study the effective conductivity tensor of a two-dimensional composite consisting of a periodic array of circular cylinders embedded in a host matrix. 
The problem has been studied before by Rayleigh \cite{R:92} for the case of a square 
array of cylinders.
His method was extended \cite{PMM:79,Mc:86} for regular arrays of circular cylinders. 
A method of functional equations \cite{Mi:97, Ryl:00} employing analytic functions was used to find an expression of the conductivity tensor for small volume fraction of inclusions. For rectangular lattice of inclusions an efficient method based on the use of elliptic functions was suggested \cite{BK:01}, which is developed further in the present paper. Effective conductivity can be also evaluated numerically \cite{M:04}.

The goal of this work is to find an analytic expression for the effective conductivity tensor in the case of an arbitrary doubly periodic array of circular cylinders when the effective tensor can be anisotropic. 
The solution of the problem consists of two steps. First, we construct a quasiperiodic potential
using a combination of Weierstrass $\zeta$-function and its derivatives. That ensures periodicity of the electric field in the whole plane and, as a result, avoids the problem of summation of a conditionally convergent series. This approach is similar to that in Ref.~\onlinecite{GF:70}. applied to biharmonic problems of the theory of elasticity. We reduce the problem to an infinite system of linear equations and find its solution in the form of a convergent power series (in terms of a parameter proportional to the volume fraction of the cylinders) whose coefficients are determined explicitly.  
Second, we determine the average electric field and the current density within one parallelogram of periods and find an exact expression of the effective conductivity tensor that relates the two quantities.

\section{Derivation of periodic potential}

Suppose that a periodic lattice of identical circular inclusions of radius $a$
is introduced into a uniform electric field $\BE = [E_x,E_y]^\tr$ applied in the plane perpendicular to the cylinder axes. The nodes of the lattice in the complex plane are generated 
by a pair of vectors $2\omo$ and $2\omt$, $\im \frac{\omt}{\omo} > 0$ (see Figure~\ref{fig1}). In polar coordinates the potential $u(r,\te)$ has the following properties: 
\begin{align}
 \Delta u =0, \quad u = \left\{
 \begin{array}{l}
  \ui \text{ in the inclusion}, \\[2mm]
  \ue \text{ in the medium},
 \end{array}
\right.
\end{align}
and on the boundary $r=a$
\begin{align}
 \ui &= \ue \label{bc1}, \\[2mm]
 \si\, \frac{\de \ui}{\de r} &= \se \,\frac{\de \ue}{\de r}, \label{bc2}
\end{align}
where $\si$ and $\se$ are the electric conductivity of the inclusions and the medium, respectively.

It is convenient to represent potential $u$ in the complex form $u = \re \,\Phi(z)$ with
\begin{align}
\label{Phi_in}
\Pi (z) &= Ea \sum_{n=0}^{\infty} \left( A_n + iB_n \right) 
\left( \frac{z}{a} \right)^{2n+1}, \quad |z| < a, \\
\Pe (z) &= \left( -E_x + iE_y \right) z 
+ Ea \sum_{n=0}^{\infty} \frac{a^{2n+1}}{(2n)!} 
\left( C_n + iD_n \right) \zeta^{(2n)}(z), \quad |z| > a,
\label{Phi_ex}
\end{align}
where $A_n, \; B_n,\; C_n,\;$ and $D_n$ are unknown real dimensionless coefficients, $E = |{\bm E}|$, and $\zeta(z)$ is Weierstrass' $\zeta$-function 
\begin{equation}
\zeta (z) = \frac{1}{z} + {\sum_{m,n}}^{\prime} \left[ \frac{1}{z-\P} + \frac{1}{\P}
+ \frac{z}{\P^2} \right].
\label{zeta} 
\end{equation}
Here $\zeta^{(2n)}(z)$ denotes derivative of order $2n$, and $\P = 2m\omo + 2n\omt$.
Prime in the sum means that summation is extended over all pairs $m,\,n$ except $m=n=0$.

\begin{figure}
\includegraphics[width=0.4\textwidth, angle=0]{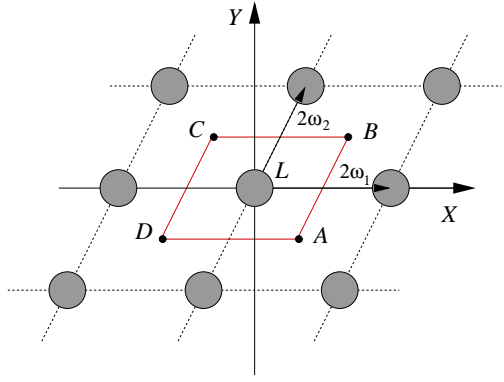}
\caption{\label{fig1} Circular inclusions of radius $a$ arranged in a
periodic lattice with periods $2\omo$ and $2\omt$.} 
\end{figure}
\noindent
Below we will use some properties of $\zeta$-function \cite{AS:65,BE:53}.
$\zeta$-function is an odd meromorphic function with simple poles at $\P$. 
It has the quasiperiodicity property
\begin{align}
 \label{q1}
 \zeta(z+2\omo) &= \zeta(z) + 2\eta_1, \quad \eta_1=\zeta(\omo), \\[2mm]
 \zeta(z+2\omt) &= \zeta(z) + 2\eta_2, \quad \eta_2=\zeta(\omt),
 \label{q2}
\end{align}
where constants $\eta_1$ and $\eta_2$ are related by the Legendre identity
\begin{equation}
 \eta_1 \,\omt - \eta_2 \,\omo = \frac{\pi i}{2}.
 \label{legendre}
\end{equation}
Derivative of $\zeta(z)$ is a periodic function and is expressed through 
Weierstrass elliptic function $\wp(z)$  by
\begin{equation}
 \zeta^\prime (z) = -\wp(z).
\end{equation}
This property ensures the electric field to be periodic in the medium, while
\rf{q1}-\rf{q2} guarantee that the potential changes by a constant value
in the direction of either $2\omo$ or $2\omt$.

Function $\wp(z)$ satisfies the differential equation
\begin{equation}
 \left[\frac{d \wp(z)}{dz}\right]^2 = 4\wp^3(z) -g_2 \,\wp(z) -g_3,
\end{equation}
where $g_2$ and $g_3$ are two invariants defined by
\begin{equation}
 g_2 = 60{\sum_{n,m}}^\prime \frac{1}{\P^{\,4}}, \quad
 g_3 = 140{\sum_{n,m}}^\prime \frac{1}{\P^{\,6}},
\end{equation}
which are used for numerical evaluation of $\wp(z)$ and $\zeta(z)$. In particular,
we will use the following homogeneity property 
\begin{equation}
 \zeta(z;g_2,g_3) = \frac{1}{\ell}\, \zeta \left( z \ell^{-1}; g_2 \ell^4, g_3 \ell^6 \right).
 \label{homo}
\end{equation}

To satisfy conditions \rf{bc1}-\rf{bc2} on the inclusion surface we expand $\zeta(z)$ 
and its even derivatives in a Laurent series 
\begin{align} 
 \zeta^{(2n)}(z) &= \frac{(2n)!}{z^{2n+1}}- \sum_{k=0}^\infty s_{n+k+1}\,
 \frac{(2n+2k+1)!}{(2k+1)!}\,z^{2k+1},
  \quad n \geq 0, \quad s_1=0,
\end{align}
where 
\begin{equation}
 s_k = {\sum_{n,m}}^\prime \frac{1}{\P^{\,2k}}, \quad k = 2,3, \ldots.
 \label{sk}
\end{equation}
Sums \rf{sk} contain only even powers of $\P$ since for every point
$\P = 2m\omo + 2n\omt$ on the lattice there exists symmetric point
$-\P$ and the sums with odd powers vanish. Also, if the periods $2\omo$ and
$2\omt$ of the lattice are fixed, sums $s_k$ remain bounded as $k \to \infty$. 
Thus, potential $u(r,\te)$ near the inclusion surface has the form
\begin{widetext}
\begin{align}
 \label{ui}
 \ui &= Ea\sum_{n=0}^\infty \left(A_n \cos[(2n+1)\te] - B_n \sin [(2n+1)\te ]\right)
 \left( \frac{r}{a} \right)^{2n+1}, \\
 \ue &= -(E_x \cos \te + E_y \sin \te)r + Ea\sum_{n=0}^\infty \frin{a}{r}{2n+1}\Bigl[ C_n \cos(2n+1)\te 
 + D_n \sin (2n+1)\te\Bigr] \nonu \\
 &-Ea \sum_{n,k=0}^\infty  \frac{(2n+2k+1)!}{(2k)!\,(2n+1)!}\,a^{2k+1}  r^{2n+1} 
  \Bigl[\bigl( C_k \cos(2n+1)\te 
  - D_k \sin (2n+1)\te\bigr)s^R_{n+k+1} \nonu \\
 &-\bigl( C_k \sin(2n+1)\te + D_k \cos (2n+1)\te\bigr)s^I_{n+k+1} \Bigr],
\end{align}
\end{widetext}
where $s^R_k$ and $s^I_k$ denote the real and imaginary parts of the sum $s_k$, respectively.

From \rf{bc1}-\rf{bc2} we obtain relations between the coefficients
\begin{align}
\label{An}
A_n &= -\frac{2\se}{\si - \se}\, C_n, \\[2mm]
\label{Bn}
B_n &= \frac{2\se}{\si - \se}\, D_n,
\end{align}
and a system for determining $C_n$ and $D_n$
\begin{widetext}
\begin{align}
 C_n -\gamma \sum_{m=0}^\infty \frac{(2n+2m+1)!}{(2m)!(2n+1)!} \left( C_m s^R_{n+m+1}
 -D_m s^I_{n+m+1}\right)a^{2n+2m+2} &= \gamma \,\frac{E_x}{E}\,\delta_{n,0}, \\[2mm]
 D_n -\gamma \sum_{m=0}^\infty \frac{(2n+2m+1)!}{(2m)!(2n+1)!} \left( -C_m s^I_{n+m+1}
 -D_m s^R_{n+m+1}\right)a^{2n+2m+2} &= \gamma \,\frac{E_y}{E}\,\delta_{n,0},
\end{align}
\end{widetext}
with $\dst \gamma = \frac{\si - \se}{\si + \se}$ and $\delta_{n,0}$ being the Kronecker delta.

Let us introduce the following notation:
\begin{align}
 \ell &= \min \{2|\omo|, 2|\omt|\}, \\[2mm]
 h &=\frac{a}{\ell}, \quad h\leq\frac{1}{2}, \\[2mm]
 S_k &= {\sum_{n,m}}^\prime \frin{\ell}{\P}{2k}=S_k^R + iS_k^I, \; k=2,3,\ldots, \; S_1 = 0.
 \label{gk}
\end{align}
Here $S_k^R$ and $S_k^I$ denote the real and imaginary part of $S_k$, respectively.
Note that once $S_2$ and $S_3$ are computed, a recurrence relation \cite{AS:65, BE:53} allows to evaluate higher order terms.
The system then can be written as
\begin{widetext}
\begin{align}
 \label{Cn}
 C_n &-\gamma \sum_{m=0}^\infty \frac{(2n+2m+1)!}{(2m)!(2n+1)!} \left( C_m S^R_{n+m+1}
 -D_m S^I_{n+m+1}\right)h^{2n+2m+2} = \gamma \,\frac{E_x}{E}\,\delta_{n,0}, \\[2mm]
 \label{Dn}
 D_n &-\gamma \sum_{m=0}^\infty \frac{(2n+2m+1)!}{(2m)!(2n+1)!} \left( -C_m S^I_{n+m+1}
 -D_m S^R_{n+m+1}\right)h^{2n+2m+2} = \gamma \,\frac{E_y}{E}\,\delta_{n,0},
\end{align}
\end{widetext}
or in vector form
\begin{align}
 {\bm x}_n - \sum_{m=0}^\infty {\bm G}_{n,m} \,{\bm x}_m h^{2n+2m+2} = {\bm y} \,\delta_{n,0},
 \label{vec_sys}
\end{align}
where
\begin{align}
 {\bm x}_n &= \left[
 \begin{array}{c}
  C_n \\[1mm]
  D_n
 \end{array}
\right],\quad
{\bm y} = \frac{\gamma}{E}\,\left[
 \begin{array}{c}
  E_x \\[1mm]
  E_y
 \end{array}
\right],\\[2mm]
 {\bm G}_{n,m} &= \gamma \,\frac{(2n+2m+1)!}{(2n+1)!(2m)!} \left[
 \begin{array}{rr}
  S^R_{n+m+1} & -S^I_{n+m+1} \\[1mm]
  -S^I_{n+m+1} & -S^R_{n+m+1}
 \end{array}
\right], \quad
 {\bm G}_{0,0} ={\bm 0}.
 \label{Gnm}
\end{align}
In Appendix we formulate conditions providing existence and  uniqueness of the solution  
of the system \rf{vec_sys} in the space of bounded sequences, as well as the possibility 
to obtain its solution by truncation or in the form of a convergent power series in $h$. 
The latter method is used below.

Let us look for solution of \rf{vec_sys} in the form of a power series in $h$
\begin{equation}
 {\bm x}_n = {\bm y} \,\delta_{n,0} + \sum_{m=0}^\infty {\bm p}_{n,m} h^{2n+2m+2}.
 \label{ser}
\end{equation}
Substituting \rf{ser} into \rf{vec_sys} and equating the coefficients of like powers of $h$
we obtain a recurrence relation for ${\bm p}_{n,m}$
\begin{align}
 {\bm p}_{n,0} &= {\bm G}_{n,0}\, {\bm y}, \\[2mm]
 {\bm p}_{n,k} &= \sum_{m=0}^{\left[ \frac{k-1}{2} \right]}{\bm G}_{n,m}\, {\bm p}_{m,k-2m-1},\quad k = 1,2,\ldots,
\end{align}
where $[s]$ denotes the integer part of $s$. Below we give several first terms of series expansion of
${\bm x}_0$ which are needed for calculation of the effective properties
\begin{widetext}
\begin{align}
{\bm x}_{0} &= {\bm y} + {\bm G}_{0,1}{\bm G}_{1,0}\,{\bm y}\, h^8
                      + {\bm G}_{0,2}{\bm G}_{2,0}\,{\bm y}\, h^{12}
                      + {\bm G}_{0,1}{\bm G}_{1,1}{\bm G}_{1,0}\,{\bm y}\, h^{14}
                      + \left(\left( {\bm G}_{0,1}{\bm G}_{1,0}\right)^2
                      + {\bm G}_{0,3}{\bm G}_{3,0}\right){\bm y}\, h^{16} \nonu\\
                      &+ \left( {\bm G}_{0,1}{\bm G}_{1,2}{\bm G}_{2,0}
                      + {\bm G}_{0,2}{\bm G}_{2,1}{\bm G}_{1,0}\right){\bm y}\, h^{18} \nonu \\
  &+ \left( {\bm G}_{0,1}{\bm G}_{1,0}{\bm G}_{0,2} {\bm G}_{2,0}
  + {\bm G}_{0,1}{\bm G}_{1,1}{\bm G}_{1,0} + 
   {\bm G}_{0,2}{\bm G}_{2,0}{\bm G}_{0,1} {\bm G}_{1,0}
   +  {\bm G}_{0,4}{\bm G}_{4,0}\right){\bm y}\, h^{20} + O(h^{22})
\label{x_0_series}
\end{align}
\end{widetext}
or
\begin{align}
{\bm x}_{0} = \M\, {\bm y},
 \label{C0D0}
\end{align}
where
\begin{widetext}
\begin{align}
 \M &=  \left(1 + 3\gamma^2 |S_2|^2 h^8
 + 5\gamma^2 |S_3|^2  h^{12} \right) \I \nonu \\[2mm]
 &+30\gamma^3
 \left[
 \begin{array}{rr}
   S_3^R\left( [S_2^R]^2 - [S_2^I]^2 \right) + 2S_2^R S_2^I S_3^I &
   S_3^I\left( [S_2^R]^2 - [S_2^I]^2 \right) - 2S_2^R S_2^I S_3^R \\[2mm]
   S_3^I\left( [S_2^R]^2 - [S_2^I]^2 \right) - 2S_2^R S_2^I S_3^R &
   -S_3^R\left( [S_2^R]^2 - [S_2^I]^2 \right) - 2S_2^R S_2^I S_3^I
 \end{array}
 \right] h^{14} \nonu \\[2mm]
 &+ \left(9\gamma^4 |S_2|^4 + 7\gamma^2 |S_4|^2 \right) h^{16} \I \nonu \\[2mm]
 +210&\gamma^3
 \left[
 \begin{array}{rr}
   S_2^R\left( S_3^R S_4^R + S_3^I S_4^I \right) + S_2^I \left( S_3^R S_4^I - S_3^I S_4^R\right) &
   -S_2^I\left( S_3^R S_4^R + S_3^I S_4^I \right) + S_2^R \left( S_3^R S_4^I - S_3^I S_4^R\right)\\[2mm]
    -S_2^I\left( S_3^R S_4^R + S_3^I S_4^I \right) + S_2^R \left( S_3^R S_4^I - S_3^I S_4^R\right)&
   S_2^I \left( S_3^I S_4^R - S_3^R S_4^I\right) - S_2^R\left( S_3^R S_4^R + S_3^I S_4^I \right)
 \end{array}
 \right] h^{18} \nonu \\[2mm]
 &+\Biggl(
 30\gamma^3
 \left[
 \begin{array}{rr}
   S_3^R\left( [S_2^R]^2 - [S_2^I]^2 \right) + 2S_2^R S_2^I S_3^I &
   S_3^I\left( [S_2^R]^2 - [S_2^I]^2 \right) - 2S_2^R S_2^I S_3^R \\[2mm]
   S_3^I\left( [S_2^R]^2 - [S_2^I]^2 \right) - 2S_2^R S_2^I S_3^R &
   -S_3^R\left( [S_2^R]^2 - [S_2^I]^2 \right) - 2S_2^R S_2^I S_3^I
 \end{array}
 \right] \nonu \\[2mm]
 &+ \left( 30\gamma^4 |S_2|^2 |S_3|^2 + 9\gamma^2 |S_5|^2 \right) \I \Biggr) h^{20}
 + O(h^{22}),
 \label{M}
\end{align}
\end{widetext}
and $\I$ is the identity matrix.

Matrix $\M$ can also be found in the form of a series expansion in $\gamma$. To this end, we rewrite
\rf{vec_sys} as
\begin{align}
 {\bm x}_n - \gamma\sum_{m=0}^\infty \widetilde{\bm G}_{n,m} \,{\bm x}_m  = \gamma \widetilde{\bm y} \,\delta_{n,0},
 \label{vec_sys1}
\end{align}
where
\begin{align}
 \label{Gtnm}
 \widetilde{\bm G}_{n,m} &=\gamma^{-1} {\bm G}_{n,m} h^{2n+2m+2}, \\[2mm]
 \widetilde{\bm y} &=\gamma^{-1} {\bm y},
\end{align}
or in operator form (see Appendix for notation)
\begin{equation}
 \left( \I -\gamma \widetilde{\bm G} \right) {\bm x} = \gamma \widetilde{\bm y}.
\end{equation}
If $\| \gamma \widetilde{\bm G} \| <1$, solution of this equation can be represented as a series
\begin{align}
 {\bm x} = \gamma \left( \I -\gamma \widetilde{\bm G} \right)^{-1} \widetilde{\bm y}
 = \sum_{n=0}^\infty \gamma^{n+1} \widetilde{\bm G}^{\,n}\, \widetilde{\bm y}.
\end{align}
From here we obtain expansion for ${\bm x}_{0}$
\begin{align}
 {\bm x}_{0} &= \gamma \widetilde{\bm y} +\gamma^3 \sum_{k=1}^\infty \widetilde{\bm G}_{0,k}
 \widetilde{\bm G}_{k,0} \widetilde{\bm y} + \gamma^4 \sum_{n,k=1}^\infty \widetilde{\bm G}_{0,n}
 \widetilde{\bm G}_{n,k} \widetilde{\bm G}_{k,0}\, \widetilde{\bm y} \nonu \\[2mm]
 &+ \gamma^5 \sum_{m,n=1}^\infty \sum_{k=0}^\infty \widetilde{\bm G}_{0,n}
 \widetilde{\bm G}_{n,k} \widetilde{\bm G}_{k,m} \widetilde{\bm G}_{m,0}\,\widetilde{\bm y}+O(\gamma^6).
\end{align}
Evaluating the second term using \rf{Gnm} and \rf{Gtnm} and comparing this expression with 
\rf{C0D0} we derive an alternative expression for $\M$
\begin{align}
 \M &= \I +\gamma^2 \I \sum_{k=1}^\infty (2k+1)|S_{k+1}|^2 h^{4k+4} + 
 \gamma^3 \sum_{n,k=1}^\infty \widetilde{\bm G}_{0,n}
 \widetilde{\bm G}_{n,k} \widetilde{\bm G}_{k,0} \nonu \\[2mm]
 &+ \gamma^4 \sum_{m,n=1}^\infty \sum_{k=0}^\infty \widetilde{\bm G}_{0,n}
 \widetilde{\bm G}_{n,k} \widetilde{\bm G}_{k,m} \widetilde{\bm G}_{m,0}
 + O(\gamma^5).
\end{align}
If all lattice sums \rf{gk} are real then expression for $\M$ is reduced to
\begin{align}
 \M &= \I +\gamma^2 \I \sum_{k=1}^\infty (2k+1)S_{k+1}^2 h^{4k+4} \nonu \\[2mm]
 &+ \gamma^3 \sum_{n,k=1}^\infty \frac{(2n+2k+1)!}{(2n)!(2k)!}\,  S_{n+1}S_{n+k+1}S_{k+1}
 \left[
 \begin{array}{rr}
  1 & 0 \\[2mm]
   0 & -1
 \end{array}
 \right] 
 h^{4n+4k+6}
  \nonu \\[2mm]
 &+ \gamma^4 \I\sum_{m,n=1}^\infty \sum_{k=0}^\infty\frac{(2n+2k+1)!(2k+2m+1)!}{(2n)!(2k)!(2k+1)!(2m)!}\,
  S_{n+1}S_{n+k+1}S_{m+k+1}S_{m+1}h^{4m+4n+4k+8} \nonu\\[2mm]
 &+ O(\gamma^5).
\end{align}

\section{Average field}

Let us find the average electric field $\langle \BE \rangle$ in the parallelogram
$ABCD$ in Figure \ref{fig1} of area $S= 4\omo \im \omt$. 
\begin{equation}
 \langle \BE \rangle = \frac{1}{S} \iint_S \BE\,dS = \frac{1}{S} \iint_{S_{in}} \BE_{in}\,dS +
 \frac{1}{S} \iint_{S_{ex}} \BE_{ex}\,dS,
 \label{E_ave}
\end{equation}
where $S_{in} = \pi a^2$ is the area of inclusion's cross-section and $S_{ex} = S - S_{in}$.
From \rf{ui} and \rf{An} we find the average field in the inclusion
in Cartesian coordinates
\begin{align}
 \langle \BE_{in} \rangle &= \frac{1}{S} \iint_{S_{in}} \BE_{in}\,dS 
 = \frac{\pi a^2 E}{S}\, [-A_0, B_0]^\tr  
 = \frac{2\pi a^2 \se E}{(\si - \se)S}\, [C_0, D_0]^\tr.
\end{align}
To calculate the average field outside the inclusion, observe that  by Green's theorem
\begin{align}
 &\iint_{S_{ex}} \BE_{ex}\,dS 
 = -\iint_{S_{ex}} \left[\frac{\de \ue}{\de x}, \frac{\de \ue}{\de y}\right]^\tr dS \nonu \\[2mm]
 &= \oint_{ABCD} \left[ -\ue\,dy, \ue\,dx \right]^\tr - \oint_{L} \left[ -\ue\,dy, \ue\,dx \right]^\tr.
\end{align}
Using the quasiperiodicity property \rf{q1}-\rf{q2} of the $\zeta$-function, one can simplify the
integrals
\begin{align}
 &\oint_{ABCD} \ue\,dy = \int_A^B \ue\,dy + \int_C^D \ue\,dy 
 = \re\int_D^C\left( -\Pe(z) + \Pe (z+2\omo)\right)dy  \nonu \\
 &=\re\left(2\omo (-E_x + iE_y) + 2Ea^2 (C_0 + iD_0)\eta_1 \right) \im(2\omt).
\end{align}
\begin{align}
 &\oint_{ABCD} \ue\,dx =  \int_{A}^B \ue\,dx + \int_{B}^C \ue\,dx+ \int_{C}^D \ue\,dx 
 + \int_{D}^A \ue\,dx\nonu \\[2mm]
 &= \re\int_{D}^C\left( \Pe(z+2\omo) - \Pe (z)\right)dx 
 -\re\int_{D}^A\left( \Pe(z+2\omt) - \Pe (z)\right)dx \nonu \\[2mm]
 &=\re\left( (-E_x + iE_y)2\omo  + 2Ea^2 (C_0 + iD_0)\eta_1 \right)\re(2\omt) \nonu \\[2mm]
 &-\re\left( (-E_x + iE_y)2\omt  + 2Ea^2 (C_0 + iD_0)\eta_2 \right)2\omo.
\end{align}
\begin{widetext}
\begin{align}
 &\oint_L \ue\,dy = \re \oint_L \Pe(z)\,dy = \re \oint_L \Pi(z)\,dy \nonu \\[2mm]
 &= \re Ea^2 \sum_{n=0}^\infty \left(A_n +iB_n \right) \int_0^{2\pi} 
 \left[ \cos (2n+1)\te +i\sin (2n+1)\te \right] \cos \te d\te = \pi a^2 E A_0.
\end{align}
\end{widetext}
Similarly,
\begin{widetext}
\begin{align}
 &\oint_L \ue\,dx = \re \oint_L \Pe(z)\,dx = \re \oint_L \Pi(z)\,dx \nonu \\[2mm]
 &= -\re Ea^2 \sum_{n=0}^\infty \left(A_n +iB_n \right) \int_0^{2\pi} 
 \left[ \cos (2n+1)\te +i\sin (2n+1)\te \right] \sin \te d\te = \pi a^2 E B_0.
\end{align}
\end{widetext}
Thus, the average electric field $\langle \BE \rangle$ has the following components: 
\begin{widetext}
\begin{align}
 \label{Ex} 
 \langle \BE \rangle_x &= E_x - \frac{Ea^2}{\omo} \re \left[ \left( C_0 +iD_0 \right) \eta_1 \right]
 = E_x - \frac{Ea^2}{\omo} \left( C_0 \re \eta_1 - D_0 \im \eta_1\right), \\[2mm]
 \langle \BE \rangle_y &= E_y + Ea^2\,\frac{\re \omt \re \left[ \left( C_0 +iD_0 \right) \eta_1 \right] - \omo \re \left[ \left( C_0 +iD_0 \right) \eta_2 \right]}{\omo \im \omt} \nonu \\[2mm]
&= E_y + Ea^2\,\frac{C_0 \im \eta_1 \im \omt - D_0 \left( \frac{\pi}{2} - \re \eta_1 \im \omt \right) }{\omo \im \omt}, 
 \label{Ey}
\end{align}
\end{widetext}
where we have used the Legendre identity \rf{legendre}.

Using expression \rf{C0D0} for the coefficients $C_0$ and $D_0$ we rewrite
\rf{Ex}-\rf{Ey} in matrix form
\begin{equation}
\langle \BE \rangle = 
\left[
 \begin{array}{l}
  \langle \BE \rangle_x \\
  \langle \BE \rangle_y
 \end{array}
 \right]= \left\{
 \I - \frac{2a^2 \gamma}{S}\, {\bm \Psi}  \M
 \right\}
 \left[
 \begin{array}{l}
  E_x \\
  E_y
 \end{array}
 \right],
\end{equation}
where 
\begin{equation}
 {\bm \Psi} = \left[
 \begin{array}{rr}
  \re \eta_1 \im 2\omt & -\im \eta_1 \im 2\omt\\[2mm]
  -\im \eta_1 \im 2\omt & \pi -\re \eta_1 \im 2\omt
 \end{array}
 \right].
 \label{Psi}
\end{equation}
Calculation of the average electric field in the inclusion and in the medium then gives
\begin{align}
 \langle \BE_{in} \rangle &= \frac{2\pi a^2 \se}{(\si + \se)S}\, \M
 \left[
 \begin{array}{l}
  E_x \\
  E_y
 \end{array}
 \right]= \frac{2\pi a^2 \se E}{(\si - \se)S}\,\x_0, \\[2mm]
 \langle \BE_{ex} \rangle &= 
 \left\{
 \I -  \frac{2\pi a^2 \se}{(\si + \se)S}\, \M - \frac{2a^2 \gamma}{S} \,{\bm \Psi} \M
 \right\}
 \left[
 \begin{array}{l}
  E_x \\
  E_y
 \end{array}
 \right] = \BE - \langle \BE_{in} \rangle - \frac{2a^2 E}{S} \,{\bm \Psi} \x_0.
\end{align}

\section{Calculation of the effective conductivity}

Effective conductivity of an arbitrary periodic array of inclusions is a tensor ${\bm \sigma}^{\ast}$,
\begin{equation}
{\bm \sigma}^{\ast} = \left[
\begin{array}{cc}
\sxe & \sxye \\
\sxye & \sye 
\end{array} \right].
\end{equation}
It relates the average current density $\langle {\bm j} \rangle$ and the average electric field $\langle \BE \rangle$
\begin{equation}
\langle{\bm j}\rangle = {\bm \sigma}^{\ast} \langle{\bm E}\rangle \quad 
\mbox{\rm or} \quad \left\{
\begin{array}{ll} 
\langle {\bm j} \rangle_x = \sxe \langle \BE\rangle_x + \sxye \langle \BE\rangle_y, &\\[2mm]
\langle {\bm j} \rangle_y = \sxye \langle \BE\rangle_x + \sye \langle \BE\rangle_y. &
\end{array} \right.
\label{j}
\end{equation}
Observe that
\begin{align}
 \langle{\bm j}\rangle &= \frac{1}{S} \iint_S \Bj\,dS = \frac{1}{S} \iint_{S_{in}} \Bj_{in}\,dS +
 \frac{1}{S} \iint_{S_{ex}} \Bj_{ex}\,dS \nonu \\[2mm]
 &= \frac{\si}{S} \iint_{S_{in}} \BE_{in}\,dS +
 \frac{\se}{S} \iint_{S_{ex}} \BE_{ex}\,dS \nonu \\[2mm]
 &= \si \langle{\BE_{in}}\rangle + \se \langle{\BE_{ex}}\rangle.
\end{align}
To determine the effective conductivity, we apply first a unit electric field in the $x$-direction:
$\BE = [E_x, E_y]^\tr = [1,0]^\tr$. The corresponding averaged vectors are denoted by the superscript $(1,0)$.
\begin{align}
 \left[
 \begin{array}{c}
 \langle {\bm j} \rangle_x^{(1,0)} \\
\langle {\bm j} \rangle_y^{(1,0)}
\end{array} \right] = 
 \left[
\begin{array}{cc}
\sxe & \sxye \\
\sxye & \sye 
\end{array} \right]
\left[
 \begin{array}{c}
 \langle {\BE} \rangle_x^{(1,0)} \\
\langle {\BE} \rangle_y^{(1,0)}
\end{array} \right].
 \label{E10}
\end{align}
Similarly, we apply then the electric field in the $y$-direction $\BE = [E_x, E_y]^\tr = [0,1]^\tr$
and compute the averaged field and current density
\begin{align}
 \left[
 \begin{array}{c}
 \langle {\bm j} \rangle_x^{(0,1)} \\
\langle {\bm j} \rangle_y^{(0,1)}
\end{array} \right] = 
 \left[
\begin{array}{cc}
\sxe & \sxye \\
\sxye & \sye 
\end{array} \right] 
\left[
 \begin{array}{c}
 \langle {\BE} \rangle_x^{(0,1)} \\
\langle {\BE} \rangle_y^{(0,1)}
\end{array} \right].
\label{E01}
\end{align}
Equation \rf{E10}-\rf{E01} can be written as one matrix equation
\begin{align}
 \left[
\begin{array}{cc}
\sxe & \sxye \\
\sxye & \sye 
\end{array} \right]
\left[
 \begin{array}{cc}
 \langle {\BE} \rangle_x^{(1,0)} & \langle {\BE} \rangle_x^{(0,1)} \\
\langle {\BE} \rangle_y^{(1,0)} & \langle {\BE} \rangle_y^{(0,1)}
\end{array} \right] = 
\left[
 \begin{array}{cc}
 \langle {\bm j} \rangle_x^{(1,0)} & \langle {\bm j} \rangle_x^{(0,1)}\\
\langle {\bm j} \rangle_y^{(1,0)} & \langle {\bm j} \rangle_y^{(0,1)}
\end{array} \right].
\end{align}
Thus,
\begin{align}
 \left[
\begin{array}{cc}
\sxe & \sxye \\
\sxye & \sye 
\end{array} \right] 
 = 
\left[
 \begin{array}{cc}
 \langle {\bm j} \rangle_x^{(1,0)} & \langle {\bm j} \rangle_x^{(0,1)}\\
\langle {\bm j} \rangle_y^{(1,0)} & \langle {\bm j} \rangle_y^{(0,1)}
\end{array} \right]
\left[
 \begin{array}{cc}
 \langle {\BE} \rangle_x^{(1,0)} & \langle {\BE} \rangle_x^{(0,1)} \\
\langle {\BE} \rangle_y^{(1,0)} & \langle {\BE} \rangle_y^{(0,1)}
\end{array} \right]^{-1}.
\label{sig1}
\end{align}
Calculation of the electric field and the current density matrices gives
\begin{align}
 \left[
 \begin{array}{cc}
 \langle {\BE} \rangle_x^{(1,0)} & \langle {\BE} \rangle_x^{(0,1)} \\
\langle {\BE} \rangle_y^{(1,0)} & \langle {\BE} \rangle_y^{(0,1)}
\end{array} \right] 
&=\I - \frac{2 \gamma a^2}{S} \,{\bm \Psi} \M, \\[2mm]
 \left[
 \begin{array}{cc}
 \langle {\bm j} \rangle_x^{(1,0)} & \langle {\bm j} \rangle_x^{(0,1)}\\
\langle {\bm j} \rangle_y^{(1,0)} & \langle {\bm j} \rangle_y^{(0,1)}
\end{array} \right] 
&=\se  \I + \frac{2 \gamma a^2 \se}{S} \left(\pi \I -\bm \Psi\right) \M.
\end{align}
Substituting these expressions in \rf{sig1} we obtain the effective conductivity tensor
\begin{equation}
 \Bs^\ast = 
 \se \left(\I+\pi \delta  \M \left(\I-\delta {\bm \Psi} \M \right)^{-1} \right),
\end{equation}
where
\begin{align}
 \delta  = \frac{2 \gamma a^2}{S}
 \label{delta}
\end{align}
is proportional to the fractional part of the inclusions.
If $\| \delta {\bm \Psi} \M \| < 1$ then $\dst \left(\I-\delta {\bm \Psi} \M \right)^{-1}$ can be
expanded in a convergent series
\begin{align}
 \Bs^\ast &= 
\se \Bigl(\I+\pi \delta \sum_{n=0}^\infty \delta^n \M \left({\bm \Psi} \M \right)^{n} \Bigr).
\label{sig2}
\end{align}
Thus, for a lattice with periods $2\omo$ and $2\omt$ (see Figure~\ref{fig1}) the expansion of the conductivity tensor in terms of volume fraction $f$ of inclusions has the form
\begin{align}
 \sigma_x^\ast &=  \se \left(1 + 2\gamma f + \frac{4\gamma^2 f^2}{\pi}\, \re \zeta(\omo) \im 2\omt
 + \frac{8\gamma^3 f^3}{\pi^2}\, | \zeta(\omo)|^2 \im^{\!2} 2\omt +O(f^4) \right), \\[2mm]
 \sigma_{xy}^\ast &= -\frac{4\se \gamma^2 f^2}{\pi}\, \im \zeta(\omo) \im 2\omt \left(1+2\gamma f
   +O(f^2) \right), \\[2mm]
 \sigma_y^\ast &= \se \left(1 + 2\gamma f + \frac{4\gamma^2 f^2}{\pi}\,\left(\pi- \re \zeta(\omo) \im 2\omt\right) \right. \nonu \\[2mm]
 &+ \left. \frac{8\gamma^3 f^3}{\pi^2}\,\left(\pi^2 -2\pi \re \zeta(\omo) \im 2\omt + | \zeta(\omo)|^2 \im^{\!2} 2\omt \right) +O(f^4)  \right),
\end{align}
where $\dst \gamma = \frac{\si - \se}{\si + \se}$.
In what follows we use series \rf{sig2} for analytic expression of the effective conductivity tensor for specific lattices.

\section{Effective conductivities of some lattices}

\subsection{Square lattice}
For the square lattice we put $2\omo = \ell,\; 2\omt=i\ell$ (see Figure \ref{fig2}). 
Then one can find \cite{AS:65}  that
\begin{align}
 \eta_1 = \frac{\pi}{2\ell}, \quad
 \eta_2 = -\frac{\pi i}{2\ell},
 \end{align}
 and from \rf{Psi} and \rf{delta} we obtain
 \begin{align}
 {\bm \Psi} &= \frac{\pi}{2} \I, \quad
 \delta = \frac{2\gamma a^2}{\ell^2}. 
\end{align}
\begin{figure}
\includegraphics[width=0.4\textwidth, angle=0]{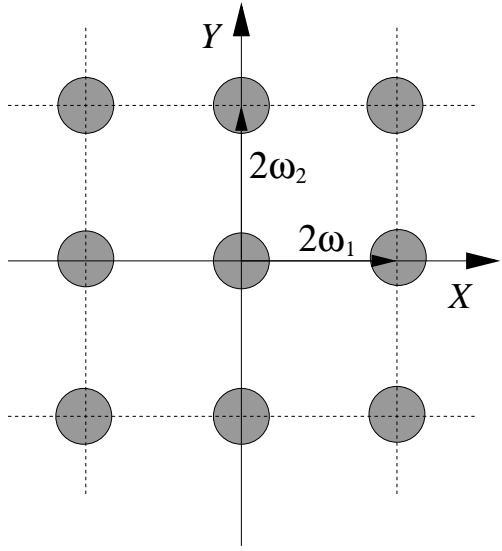}
\caption{Square lattice of inclusions of radii $a$ with periods $2\omo=\ell$ and $2\omt=i\ell$.} 
\label{fig2}
\end{figure}
All lattice sums \rf{gk} are real with the only non-zero being $S_{2k}$, $k=1,2,\ldots$.
Substituting these parameters into \rf{sig2} we obtain the effective conductivity tensor of the square lattice
\begin{equation}
 \Bs^\ast = \se \Bigl\{ \I + 2\sum_{n=0}^\infty \left( \gamma f \M \right)^{n+1} \Bigr\},
 \label{sig3}
\end{equation}
where $\dst f=\frac{\pi a^2}{\ell^2}$ is the volume fraction of the inclusions. Calculation of matrix $\M$ in \rf{M} gives
\begin{equation}
 \M = \left( 1 + 3\gamma^2 S_2^2 h^8 + \left(9\gamma^4 S_2^4 + 7\gamma^2 S_4^2 \right)h^{16} \right) \I + O(h^{24}).
\end{equation}
Effective conductivity tensor of the square lattice is isotropic, $\Bs^\ast = \sigma^\ast \I$, and for
$\sigma^\ast$ we obtain from \rf{sig3}
\begin{align}
 \sigma^\ast &=\se \Biggl(1+2\gamma f+2\gamma^2 f^2+2\gamma^3 f^3+2\gamma^4 f^4 \nonu\\[2mm]
 &+2\left( \gamma^5 +\frac{3 \gamma^3 S_2^2}{\pi^4} \right) f^5 
 +2\left( \gamma^6 +\frac{6\gamma^4 S_2^2}{\pi^4} \right) f^6 \nonu\\[2mm]
 &+2\left( \gamma^7 +\frac{9\gamma^5 S_2^2}{\pi^4} \right) f^7
 +2\left( \gamma^8 +\frac{12\gamma^6 S_2^2}{\pi^4} \right) f^8 \nonu \\[2mm]
 &+2\left( \gamma^9 +\frac{9\gamma^5 S_2^4 + 15\pi^4 \gamma^7 S_2^2 + 7\gamma^3 S_4^2}{\pi^8}\right) f^9 \nonu \\[2mm]
 &+2\left( \gamma^{10} +\frac{27\gamma^6 S_2^4 + 18\pi^4 \gamma^8 S_2^2 + 14\gamma^4 S_4^2}{\pi^8}\right) f^{10} \nonu \\[2mm]
 &+2\left( \gamma^{11} +\frac{54\gamma^7 S_2^4 + 21\pi^4 \gamma^9 S_2^2 + 21\gamma^5 S_4^2}{\pi^8}\right) f^{11} \nonu \\[2mm]
 &+2\left( \gamma^{12} +\frac{90\gamma^8 S_2^4 + 24\pi^4 \gamma^{10} S_2^2 + 28\gamma^6 S_4^2}{\pi^8}\right) f^{12} \nonu\\[2mm]
 &+ O(f^{13}).
 \label{square}
\end{align}
Here $\dst S_2 = {\sum_{n,m}}^\prime \frac{1}{(m+in)^4} = 3.15121$, $\dst S_4 = {\sum_{n,m}}^\prime \frac{1}{(m+in)^8}=4.25577$ correct to five decimal places.
Expression \rf{square} is in agreement with known results \cite{R:92,PMM:79}.

\subsection{Regular triangular lattice}

The effective conductivity tensor is also isotropic in the case of a regular triangular lattice (see Figure \ref{fig3}).
Similar to the previous case  we put $\dst 2\omo = \ell,\; 2\omt=\ell e^{\pi i/3}$. Then we find \cite{AS:65} that
\begin{align}
 \eta_1 &= \frac{\pi}{\ell \sqrt{3}}, \quad
 \eta_2 = \frac{\pi e^{\pi i/3}}{\ell \sqrt{3}},
\end{align}
and as a result,
\begin{align}
 {\bm \Psi} &= \frac{\pi}{2} \I, \quad
 \delta = \frac{2\gamma a^2}{S}. 
\end{align}
\begin{figure}
\includegraphics[width=0.6\textwidth, angle=0]{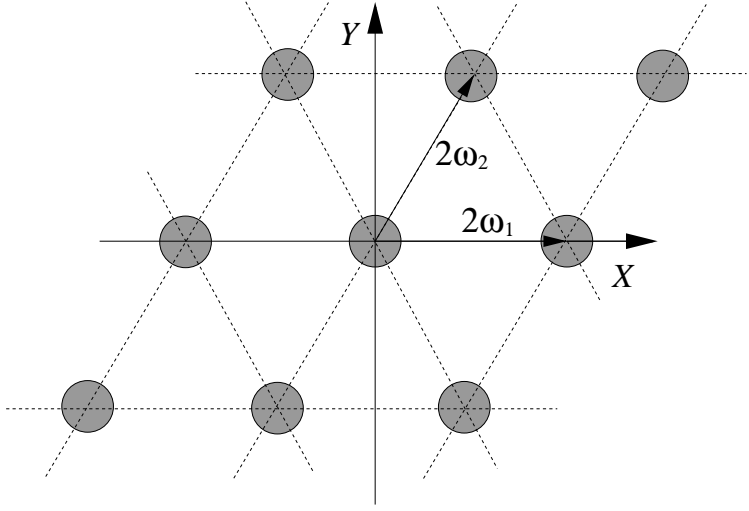}
\caption{Regular triangular lattice of inclusions of radii $a$ with periods $2\omo=\ell$ and $2\omt=\ell e^{i\pi/3}$.} 
\label{fig3}
\end{figure}
All lattice sums \rf{gk} are real with the only non-zero being $S_{3k}$, $k=1,2,\ldots$.
Substituting these parameters into \rf{sig2} we obtain the effective conductivity tensor of the regular triangular lattice
\begin{equation}
 \Bs^\ast = \se \Bigl\{ \I + 2\sum_{n=0}^\infty \left( \gamma f \M \right)^{n+1} \Bigr\},
 \label{sig4}
\end{equation}
where $\dst f=\frac{\pi a^2}{S}$ is the fractional part of the inclusions. Matrix $\M$ found
from \rf{M} is
\begin{equation}
 \M (h) = \left( 1+5\gamma^2 S_3^2 h^{12} \right) \I + O \left(h^{24}\right).
 \label{}
\end{equation}

The effective conductivity tensor of the regular triangular lattice is isotropic, $\Bs^\ast = \sigma^\ast \I$, and for
$\sigma^\ast$ we obtain from \rf{sig4}
\begin{align}
 \sigma^\ast &=\se \Biggl(1+2\gamma f+2\gamma^2 f^2+2\gamma^3 f^3+2\gamma^4 f^4
 +2\gamma^5 f^5 \nonu \\[2mm]
 &+2\gamma^6 f^6 
 +\left( 2\gamma^7 +\frac{135\gamma^3 S_3^2}{32\pi^6} \right) f^7 \nonu \\[2mm]
 &+\left( 2\gamma^8 +\frac{135\gamma^4 S_3^2}{16\pi^6} \right) f^8 
 +\left( 2\gamma^9 +\frac{405\gamma^5 S_3^2}{32\pi^6}\right) f^9 \nonu \\[2mm]
 &+\left( 2\gamma^{10} +\frac{135\gamma^6 S_3^2}{8\pi^6}\right) f^{10} 
 +\left( 2\gamma^{11} +\frac{675\gamma^7 S_3^2}{32\pi^6}\right) f^{11} \nonu \\[2mm]
 &+ \left( 2\gamma^{12} +\frac{405\gamma^8 S_3^2}{16\pi^6}\right) f^{12} \Biggr)
 + O\left(f^{13}\right). 
\end{align}
The latter formula agrees with calculation in Ref. \onlinecite{PMM:79}.
Here $\dst S_3 = {\sum_{n,m}}^\prime \frac{1}{\left(m+ne^\frac{\pi i}{3}\right)^6}=5.86303$ correct to five decimal places.

\subsection{Rectangular lattice}

Consider a rectangular lattice generated by the vectors $2\omo = 2\ell,\; 2\omt=i\ell$. We compute the lattice sums
\begin{align}
 S_2 & = {\sum_{n,m}}^\prime \frac{\ell}{\left(2\omo m+2\omt n\right)^4}
 = {\sum_{n,m}}^\prime \frac{1}{\left(2m+ in\right)^4} = 2.16646, \\[2mm] 
 S_3 &= {\sum_{n,m}}^\prime \frac{1}{\left(2m+ in \right)^6} =-2.03111.
\end{align}
\begin{figure}
\includegraphics[width=0.6\textwidth, angle=0]{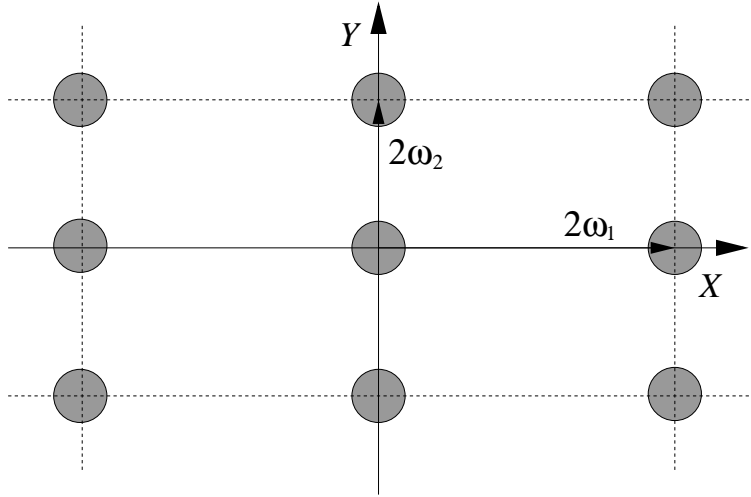}
\caption{Rectangular lattice of inclusions of radii $a$ with periods $2\omo=2\ell$ and $2\omt=i\ell$.} 
\label{fig4}
\end{figure}
Next, we compute the invariants $g_2$ and $g_3$:
\begin{align}
 g_2 &= 60\,S_2 \ell^{-4}= 129.988\ell^{-4}, \\[2mm]
 g_3 &= 140\,S_3\ell^{-6} = -284.355\ell^{-6}.
\end{align}
Then we compute $\eta_1$ from \rf{q1} using \rf{homo}
\begin{align}
 \eta_1 &= \zeta(\omo; g_2,g_3) = \ell^{-1} \zeta \left(1; 129.988, -284.355\right) \nonu \\[2mm]
 & =-0.14800 \ell^{-1} 
\end{align}
and matrices ${\bm \Psi}$ and $\M$
\begin{align}
 {\bm \Psi} &= \left[
 \begin{array}{cc}
  -0.14800 & 0 \\[2mm]
  0 & 3.28959
 \end{array}
 \right],
\end{align}
\begin{align}
 \M &= \left(1 +3\gamma^2 S_2^2 h^8 + 5\alpha^2 S_3^2 h^{12}\right) \I
 + 30 \gamma^3 S_2^2 S_3 h^{14} 
 \left[
 \begin{array}{cc}
  1 & 0 \\[2mm]
  0 & -1
 \end{array}
 \right] + O(h^{16}).
\end{align}
From \rf{sig2} we obtain that the effective conductivity tensor ${\bm \sigma}^\ast$ is diagonal 
\begin{align}
 {\bm \sigma}^\ast = \left[
 \begin{array}{cc}
  \sigma_x^\ast & 0 \\[2mm]
  0 & \sigma_y^\ast
 \end{array}
\right]
\end{align}
with components
\begin{align}
 \sigma_x^\ast &= \sigma_{ex} \left(1 + 2\gamma f -0.188439 \gamma^2 f^2 +0.0177547 \gamma^3 f^3 -0.00167284 \gamma^4 f^4 + O \left(f^5 \right)  \right), \\[2mm]
 \sigma_y^\ast &= \sigma_{ex} \left( 1 + 2 \gamma f +4.18844 \gamma^2 f^2 + 8.77150 \gamma^3 f^3 +18.3694 \gamma^4 f^4 + O \left(f^5 \right)  \right),
\end{align}
where $f$ is the volume fraction of the cylinders $f=\frac{1}{2}\pi h^2$.

\subsection{Anisotropic lattice}

Here we show how to find the effective conductivity tensor for an arbitrary lattice. 
Consider the case when the lattice is created by the vectors $2\omo = 2\ell,\; 2\omt=\ell e^{\pi i/3}$ (see Figure \ref{fig5}). Then we calculate the lattice sums
\begin{align}
 S_2 & = {\sum_{n,m}}^\prime \frac{\ell^4}{\left(2\omo m+2\omt n\right)^4}
 = {\sum_{n,m}}^\prime \frac{1}{\left(2m+ e^\frac{\pi i}{3}n\right)^4} =
 -1.08720+1.88309i, \\[2mm] 
 S_3 &= {\sum_{n,m}}^\prime \frac{\ell^6}{\left(2m+ e^\frac{\pi i}{3}n\right)^6} = 2.01542,
\end{align}
\begin{figure}
\includegraphics[width=0.75\textwidth, angle=0]{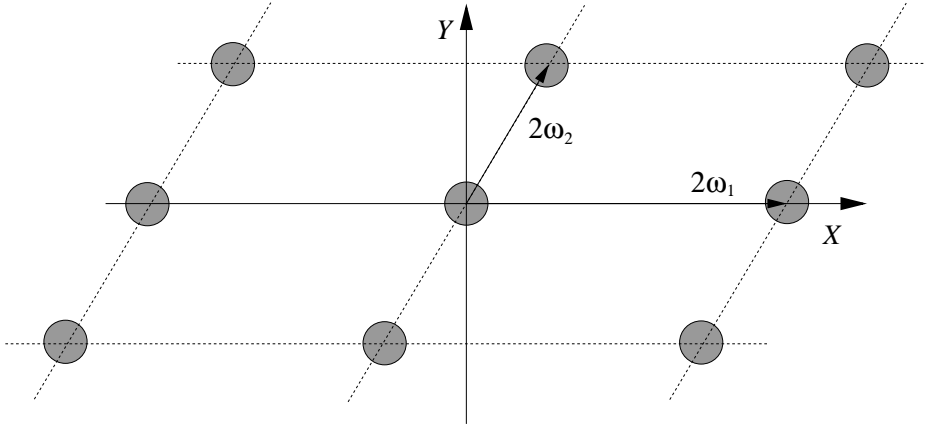}
\caption{Generic lattice of inclusions of radii $a$ with periods $2\omo=2\ell$ and $2\omt=\ell e^{\pi i/3}$.} 
\label{fig5}
\end{figure}
and the invariants
\begin{align}
 g_2 &= 60 S_2\,\ell^{-4} =\left( -65.2321 + 112.985 i \right) \ell^{-4},\\[2mm]
 g_3 &= 140 S_3 \,\ell^{-6}= 282.158 \, \ell^{-6}.
\end{align}
In order to find constant $\eta_1$ we use the homogeneity property of the $\zeta$-function \rf{homo}
\begin{align}
 \eta_1 &= \zeta(\omo, g_2,g_3) = \ell^{-1}\zeta \left(1, -65.2321 + 112.985 i, 282.158 \right) \nonu \\[2mm]
 & = (1.07651-1.27703i)\ell^{-1}.
\end{align}
Now we evaluate matrix ${\bm \Psi}$ in \rf{Psi}
\begin{align}
 {\bm \Psi} &= \left[
 \begin{array}{rr}
  \re \eta_1 \im 2\omt & -\im \eta_1 \im 2\omt\\[2mm]
  -\im \eta_1 \im 2\omt & \pi -\re \eta_1 \im 2\omt
 \end{array}
 \right] =  \left[
 \begin{array}{rr}
  0.93228  & 1.10594\\[2mm]
  1.10594 & 2.20931
 \end{array}
 \right],
\end{align}
and using the expansion of $\M(h)$ in \rf{M} we compute from \rf{sig2} the effective conductivity tensor 
\begin{align}
 {\bm \sigma}^\ast = \left[
 \begin{array}{cc}
  \sigma_x^\ast & \sigma_{xy}^\ast \\[2mm]
  \sigma_{xy}^\ast & \sigma_y^\ast
 \end{array}
\right],
\end{align}
where
\begin{align}
 \sigma_{x}^\ast &= \sigma_{ex} \left(1 + 2\gamma f +1.18702 \gamma^2 f^2 +1.69592 \gamma^3 f^3 + 2.98936 \gamma^4 f^4 + O \left(f^5 \right)  \right), \\[2mm]
 \sigma_{xy}^\ast &= \sigma_{ex} \left(1.40813 \gamma^2 f^2 +2.81625 \gamma^3 f^3 + 5.15506 \gamma^4 f^4 + O \left(f^5 \right)  \right), \\[2mm]
 \sigma_y^\ast &= \sigma_{ex} \left( 1 + 2 \gamma f +2.81298 \gamma^2 f^2 + 4.94784 \gamma^3 f^3 + 8.94191 \gamma^4 f^4 + O \left(f^5 \right)  \right). 
\end{align}

\section*{Appendix}
\setcounter{equation}{0}
\renewcommand{\theequation}{A\arabic{equation}}

Here we study properties of system \rf{vec_sys} using the approach similar to that in Ref.~\onlinecite{GZ88}. 
We seek a solution of \rf{vec_sys} in the space of bounded sequences $l_{\infty} ({\R}^{2})$ whose elements are two-dimensional vectors
$${\bm x}=\{{\bm x}_{n}\}_{n=0}^{\infty}=\left[{\bm x}_{0}, {\bm x}_{1},\ldots,
{\bm x}_{n},\ldots\right]^\tr,$$
where ${\bm x}_{n}=[\scal{x}{n}{1},\scal{x}{n}{2}]^\tr$.
\noindent
The norm of an element of this space is given by
$$\|\bm{x}\| =\sup_{n=0,1,2,\ldots}\;\max
\left\{ | \scal{x}{n}{1} |, | \scal{x}{n}{2} |
\right\}. $$
We introduce a linear operator ${\G}\,(h)$ by
\begin{equation}
 \left({\G}\,(h) \x\right)_n = \sum_{m=0}^\infty  {\bm G}_{n,m}\,{\x}_m \,h^{2n+2m+2},
 \quad n = 0,1,2,\ldots.
 \label{Gop}
\end{equation}
Then \rf{vec_sys} can be written in operator form
\begin{equation}
 \x - {\G}\,(h) \x = {\bm y}.
 \label{oper}
\end{equation}
Properties of operator $\G\,(h)$ and equation \rf{oper} are summarized in the following
\begin{theorem}
For each $0\leq h \leq \frac{1}{2}$ $\G\,(h)$ is a bounded operator in $l_{\infty} ({\R}^{2})$.
If $0\leq h < \frac{1}{2}$ then $\G\,(h)$ is compact and can be represented by a convergent series
\begin{equation}
 \G\,(h) = \sum_{m=1}^\infty h^{2m} {\bm G}^{(m)},
 \label{Gexp}
\end{equation}
where ${\bm G}^{(m)}$ are finite-dimensional operators of order $2m$.
\end{theorem}
\begin{proof}
Let us estimate the norm of $\G\, (h)$: 
\begin{align}
 \|\G\, (h) \| &= \sup_{\substack{\| \x \| \leq 1 \\ \x \neq {\bm 0}}} \frac{\|\G\, (h) \,\x \|}{\|\x\|} = \sup_{\substack{\| \x \| \leq 1 \\ \x \neq {\bm 0}}} \sup_{n}\frac{\| \sum_{m=0}^\infty  {\bm G}_{n,m}\,{\x}_m \,h^{2n+2m+2} \|}{\|\x\|} \nonu \\[2mm]
 & \leq \sup_{n} \sum_{m=0}^\infty \| {\bm G}_{n,m}\| \,h^{2n+2m+2} 
 \leq |\gamma| \sup_{n} \left( |S_n^R| + |S_n^I| \right) \nonu \\[2mm]
 &\times \sup_{n} \sum_{m=0}^\infty
 \frac{(2n+2m+1)!}{(2m)!\,(2n+1)!}\, h^{2n+2m+2} 
 = |\gamma|\,  \widetilde{S}\, \sup_{n}
 \frac{ h^{2n+2}}{(2n+1)!} \frac{d^{2n+1}}{dh^{2n+1}} \sum_{m=0}^\infty h^{2m+2n+1} \nonu \\[2mm]
 &= |\gamma|  \,  \widetilde{S}\, \sup_{n} \frac{ h^{2n+2}}{(2n+1)!}
 \frac{d^{2n+1}}{dh^{2n+1}} \left(-\frac{1}{1-h}-\frac{1}{1+h} \right) \nonu \\[2mm]
 &=|\gamma| \,  \widetilde{S}\, \sup_{n}
 \left(\left(\frac{h}{1-h}\right)^{2n+2}+\left(\frac{h}{1+h}\right)^{2n+2} \right) \nonu \\[2mm]
 &= |\gamma| \,  \widetilde{S}\left(\left(\frac{h}{1-h}\right)^{4}+\left(\frac{h}{1+h}\right)^{4} \right) < \infty,
 \label{Gnorm}
\end{align}
where $\widetilde{S} = \sup_{n} \left( |S_n^R| + |S_n^I| \right)$. Therefore $\G\,(h)$ is a bounded operator for $0 \leq h \leq \frac{1}{2}$. From \rf{Gnorm} it
also follows that if $0 \leq h < \frac{1}{2}$ than $\G\,(h)$ maps a bounded sequence into the space $c_0 (\R^2)$ of sequences converging to zero and hence it is compact.

Expansion \rf{Gexp} follows formally from the definition \rf{Gop} of operator $\G\,(h)$,
where $2m$-dimensional operators ${\bm G}^{(m)}$ are defined by
\begin{equation}
 \left( {\bm G}^{(m)} {\bm x} \right)_k = \left\{
 \begin{array}{ll}
  {\bm G}_{m-k-1,k} \, \x_k, & 0 \leq k \leq m-1, \\[2mm]
  0, & k \geq m.
 \end{array}
\right.
\end{equation}
To show convergence of the series \rf{Gexp} we observe that
\begin{align}
 \|  {\bm G}^{(m)} \| &\leq \max_{0\leq k \leq m-1} \| {\bm G}_{m-k-1,k} \|
 = |\gamma| \frac{(2m-1)!}{(2k)!\,(2m-2k-1)!} \left( \left| S_m^R \right| + \left| S_m^I \right|
 \right) \nonu \\[2mm]
 & \leq |\gamma| \widetilde{S}\, \frac{(2m-1)!}{m!\,(m-1)!}.
\end{align}
Therefore series \rf{Gexp} is dominated by a convergent for $0 \leq h < \frac{1}{2}$ series
\begin{align}
 \sum_{m=1}^\infty \|{\bm G}^{(m)}\|  h^{2m} \leq |\gamma| \widetilde{S} \sum_{m=1}^\infty \frac{(2m-1)!}{m!\,(m-1)!} h^{2m}=\frac{2|\gamma|h^2 \widetilde{S}}{\sqrt{1-4h^2} \left( 1+\sqrt{1-4h^2} \right)}.
\end{align}

For values of $\gamma$, $\widetilde{S}$, and $h$ such that $\| \G\, (h)\| < 1$ in \rf{Gnorm}
the fixed point theorem for contraction operators on Banach spaces ensures the following properties
of the solution of \rf{vec_sys}:
\begin{itemize}
 \item[(a)] Equation \rf{vec_sys} has a unique solution ${\bm x}_0 \in c_0 (\R^2)$.
 \item[(b)] The truncated solution of \rf{vec_sys} converges exponentially to ${\bm x}_0$.
 \item[(c)] The solution of \rf{vec_sys} can be represented as a convergent power series in $h$.
\end{itemize}
\end{proof}

\end{document}